\documentclass[11pt]{article}
\usepackage{amsmath}
\usepackage{times}
\usepackage{amssymb}
\usepackage{amsthm}
\usepackage{fullpage}
\usepackage{epsfig}
\usepackage{subfigure}
\usepackage{calc}

\newtheorem{theorem}{Theorem}
\newtheorem{conjecture}{Conjecture}
\newtheorem{claim}{Claim}

\newtheorem{definition}{Definition}
\newtheorem{observation}{Observation}
\newtheorem{proposition}{Proposition}

\newtheorem{lemma}[theorem]{Lemma}

\theoremstyle{definition}

\newtheorem{remark}{Remark}
\newcommand{\expct}[2]{\mathbb{E}_{#1}\left[#2\right]}
\newcommand{\prob}[2]{\mathsf{Pr}_{#1}\left[#2\right]}
\newcommand{\var}[2]{\mathsf{Var}_{#1}\left[#2\right]}

\newcommand{\R}{\mathbb{R}}

\newcommand{\sugg}{\mathsf{Sugg}}
\newcommand{\opt}{\mathsf{Opt}}
\newcommand{\dtoone}[1]{{\sf ${#1}$-to-$1$}}

\newcommand{\xiffy}[2]{\ensuremath{{#1} \mathcal{\leftrightarrow} {#2}}}

\newcommand{\intset}[1]{\ensuremath{\{1,\ldots,{#1}\}}}
\newcommand{\infl}{\ensuremath{\mathsf{Inf}}}
\newcommand{\maxcolor}[1]{{\sf Max ${#1}$-Colorable Subgraph}}
\newcommand{\maxcut}[1]{{\sf Max ${#1}$-Cut}}
\newcommand{\mcut}{{\sf MaxCut}}
\newcommand{\eps}{\varepsilon}

\newcommand{\texsup}[2]{\ensuremath{\text{#1}^{\text{#2}}}}

\renewcommand{\epsilon}{\varepsilon}
\renewcommand{\ge}{\geqslant}

\renewcommand{\le}{\leqslant}

\newcommand{\threehardness}{33}
\newcommand{\threehardnessdec}{32}

\title{{\bf Improved Inapproximability Results for \\ Maximum $k$-Colorable Subgraph}}

\author{{\sc Venkatesan Guruswami}\footnote{Research supported in part
    by a Packard Fellowship. Email: {\tt guruswami@cmu.edu}, {\tt asinop@cs.cmu.edu}} \and {\sc Ali Kemal
    Sinop}\footnotemark[\value{footnote}]}

\date{Computer Science Department \\  Carnegie Mellon University \\ Pittsburgh, PA 15213.} 
\begin{document}

\maketitle
\thispagestyle{empty}

\begin{abstract}
  We study the maximization version of the fundamental graph coloring
  problem. Here the goal is to color the vertices of a $k$-colorable
  graph with $k$ colors so that a maximum fraction of edges are
  properly colored (i.e. their endpoints receive different colors). A
  random $k$-coloring properly colors an expected fraction
  $1-\frac{1}{k}$ of edges. We prove that given a graph promised to be
  $k$-colorable, it is NP-hard to find a $k$-coloring that properly
  colors more than a fraction $\approx 1-\frac{1}{\threehardness k}$
  of edges. Previously, only a hardness factor
  of $1- O\bigl(\frac{1}{k^2}\bigr)$ was known. Our result pins down
  the correct asymptotic dependence of the approximation factor on
  $k$. Along the way, we prove that approximating the Maximum
  $3$-colorable subgraph problem within a factor greater than $\frac{32}{33}$
  is NP-hard.

  Using semidefinite programming, it is known that one can do better
  than a random coloring and properly color a fraction $1-\frac{1}{k}
  +\frac{2 \ln k}{k^2}$ of edges in polynomial time. We show that,
  assuming the $2$-to-$1$ conjecture, it is hard to properly color
  (using $k$ colors) more than a fraction $1-\frac{1}{k} + O\left(\frac{\ln
    k}{k^2}\right)$ of edges of a
  $k$-colorable graph.

\end{abstract}

\newpage

\section{Introduction}
\vspace{-1ex}
\subsection{Problem statement}
A graph $G=(V,E)$ is said to be $k$-colorable for some positive
integer $k$ if there exists a $k$-coloring $\chi: V \rightarrow
\{1,2,\dots,k\}$ such that for all edges $(u,v) \in E$, $\chi(u) \neq
\chi(v)$. For $k \ge 3$, finding a $k$-coloring of a $k$-colorable
graph is a classic NP-hard problem.  The problem of coloring a graph
with the fewest number of colors has been extensively studied.  In
this paper, our focus is on hardness results for the following
maximization version of graph coloring: Given a $k$-colorable graph
(for some fixed constant $k \ge 3$), find a $k$-coloring that
maximizes the fraction of properly colored edge. (We say an edge is
properly colored under a coloring if its endpoints receive distinct
colors.) Note that for $k=2$ the problem is trivial --- one can find a
proper $2$-coloring in polynomial time when the graph is bipartite
($2$-colorable).

We will call this problem \maxcolor{k}. The problem is equivalent to
partitioning the vertices into $k$ parts so that a maximum number of
edges are cut. This problem is more popularly referred to as
\maxcut{k} in the literature; however, in the \maxcut{k} problem the
input is an arbitrary graph that need not be $k$-colorable. To
highlight this difference that our focus is on the case when the input
graph is $k$-colorable, we use \maxcolor{k} to refer to this
variant. We stress that we will use this convention throughout the
paper: \maxcolor{k} {\bf always} {\em refers to the ``perfect
  completeness'' case, when the input graph is
  $k$-colorable.}\footnote{While a little non-standard, this makes our
  terminology more crisp, as we can avoid repeating the fact that the
  hardness holds for $k$-colorable graphs in our statements.} Since
our focus is on hardness results, we note that this restriction only
makes our results stronger.

A factor $\alpha=\alpha_k$ approximation algorithm for \maxcolor{k} is
an efficient algorithm that given as input a $k$-colorable graph
outputs a $k$-coloring that properly colors at least a fraction
$\alpha$ of the edges. We say that \maxcolor{k} is NP-hard to
approximate within a factor $\beta$ if no factor $\beta$ approximation
algorithm exists for the problem unless ${\rm P} = {\rm NP}$. The goal
is to determine the approximation threshold of \maxcolor{k}: the
largest $\alpha$ as a function of $k$ for which a factor $\alpha$
approximation algorithm for \maxcolor{k} exists.
\vspace{-1ex}
\subsection{Previous results}
The algorithm which simply picks a random $k$-coloring, without even
looking at the graph, properly colors an expected fraction $1-1/k$ of
edges. Frieze and Jerrum~\cite{FJ97} used semidefinite programming to
give a polynomial time factor $1-1/k + 2 \ln k/k^2$ approximation
algorithm for \maxcut{k}, which in particular means the algorithm will
color at least this fraction of edges in a $k$-colorable graph. This
remains the best known approximation guarantee for \maxcolor{k}\ to
date. Khot, Kindler, Mossel, and O'Donnell~\cite{KKMO} showed that
obtaining an approximation factor of $1-1/k + 2 \ln
k/k^2 + \Omega(\ln\ln k/k^2)$ for \maxcut{k}\ is Unique Games-hard,
thus showing that the Frieze-Jerrum algorithm is essentially the best
possible. However, due to the ``imperfect completeness'' inherent to
the Unique Games conjecture, this hardness result does {\em not} hold
for \maxcolor{k} when the input is required to be $k$-colorable.

For \maxcolor{k}, the best hardness known prior to our work was a
factor $1-\Theta(1/k^2)$. This is obtained by combining an
inapproximability result for \maxcolor{3} due to
Petrank~\cite{petrank} with a reduction from Papadimitriou and
Yannakakis~\cite{PY91}. It is a natural question whether is an
efficient algorithm that could properly color a fraction $1-
1/k^{1+\eps}$ of edges given a $k$-colorable graph for some absolute
constant $\eps > 0$. The existing hardness results do not rule out the
possibility of such an algorithm.

For \maxcut{k}, a better hardness factor was shown by Kann, Khanna, Lagergren, and Panconesi~\cite{KKLP} ---
for some absolute constants $\beta > \alpha > 0$, they
showed that it is NP-hard to distinguish graphs that have a $k$-cut in
which a fraction $(1-\alpha/k)$ of the edges cross the cut from graphs whose
Max $k$-cut value is at most a fraction $(1-\beta/k)$ of edges.
Since \mcut\ is easy when the graph is $2$-colorable, this reduction
does not yield any hardness for \maxcolor{k}.
\vspace{-1ex}
\subsection{Our results}
Petrank~\cite{petrank} showed the existence of a $\gamma_0 > 0$ such that
it is NP-hard to find a $3$-coloring that properly colors more than a
fraction $(1-\gamma_0)$ of the edges of a $3$-colorable graph. The
value of $\gamma_0$ in \cite{petrank} was left unspecified and would
be very small if calculated. The reduction in \cite{petrank} was rather
complicated, involving expander graphs and starting from the weak
hardness bounds for bounded occurrence satisfiability. We prove that
the NP-hardness holds with $\gamma_0 = \frac{1}{\threehardness}$. In
other words, it is NP-hard to obtain an approximation ratio bigger
than $\frac{\threehardnessdec}{\threehardness}$ for \maxcolor{3}. The
reduction is from the constraint satisfaction problem corresponding to
the adaptive $3$-query PCP with perfect completeness from \cite{GLST}.

By a reduction from \maxcolor{3}, we prove that for every $k \ge 3$,
the \maxcolor{k} is NP-hard to approximate within a factor greater
than $\approx 1-\frac{1}{\threehardness k}$
(Theorem~\ref{thm:kcolor-NP-hard}).  This identifies the correct
asymptotic dependence on $k$ of the best possible approximation factor
for \maxcolor{k}. The reduction is similar to the one in \cite{KKLP},
though some crucial changes have to be made in the construction and
some new difficulties overcome in the soundness analysis when reducing
from \maxcolor{3} instead of \mcut.


In the quest for pinning down the {\em exact} approximability of
\maxcolor{k}, we prove the following {\em conditional}
result. Assuming the so-called $2$-to-$1$ conjecture, it is hard to
approximate \maxcolor{k} within a factor $1-\frac{1}{k} + O\left(\frac{\ln
  k}{k^2}\right)$. In other words, the Frieze-Jerrum algorithm is optimal up to
lower order terms in the approximation ratio {\em even for instances of
\maxcut{k} where the graph is $k$-colorable.} %

Unlike the Unique Games Conjecture (UGC), the $2$-to-$1$ conjecture
allows perfect completeness, i.e., the hardness holds even for
instances where an assignment satisfying {\em all} constraints
exists. The $2$-to-$1$ conjecture was used by Dinur, Mossel, and
Regev~\cite{DMR} to prove that for every constant $c$, it is NP-hard
to color a $4$-colorable graph with $c$ colors. We analyze a similar
reduction for the $k$-coloring case when the objective is to maximize
the fraction of edges that are properly colored by a $k$-coloring. Our
analysis uses some of the machinery developed in \cite{DMR}, which in
turn extends the invariance principle of \cite{MOO}. The hardness
factor we obtain depends on the spectral gap of a certain $k^2 \times
k^2$ stochastic matrix. 

\begin{remark}
In general it is far from clear which Unique
Games-hardness results can be extended to hold with perfect
completeness by assuming, say, the $2$-to-$1$ (or some related)
conjecture. In this vein, we also mention the result of O'Donnell and
Wu~\cite{odonnell-wu} who showed a tight hardness for approximating
satisfiable constraint satisfaction problems on $3$ Boolean variables
assuming the $d$-to-$1$ conjecture for any fixed $d$. While the UGC
assumption has led to a nearly complete understanding of the
approximability of constraint satisfaction problems~\cite{Raghavendra08}, the
approximability of {\em satisfiable} constraint satisfaction problems
remains a mystery to understand in any generality.
\end{remark}

\begin{remark}
\label{rem:wt-unwt}
  It has been shown by Crescenzi, Silvestri and Trevisan \cite{CST01}
  that any hardness result for weighted instances of \maxcut{k}
  carries over to  unweighted instances {assuming} the
  total edge weight is polynomially bounded. In fact, their reduction
  preserves $k$-colorability, so an inapproximability result for the weighted
  \maxcolor{k} problem also holds for the unweighted version. Therefore
  all our hardness results hold for the unweighted \maxcolor{k} problem.
\end{remark}

\section{Unconditional Hardness Results for \maxcolor{k}}
We will first prove a hardness result for \maxcolor{3}, and then
reduce this problem to \maxcolor{k}.
\subsection{Inapproximability result for \maxcolor{3}}
\label{sec:3color-hard}
Petrank~\cite{petrank} showed that \maxcolor{3} is NP-hard to
approximate within a factor of $(1-\gamma_0)$ for some constant
$\gamma_0 > 0$. This constant $\gamma_0$ is presumably very small,
since the reduction starts from bounded occurrence satisfiability (for
which only weak inapproximability results are known) and uses expander
graphs. We prove a much better inapproximability factor below, via a
simpler proof.


\begin{theorem}[\maxcolor{3} Hardness]
\label{thm:3color}
  The \maxcolor{3} problem is NP-hard to approximate within a factor of
  $\frac{\threehardnessdec}{\threehardness}+\eps$ for any constant $\eps > 0$.
\end{theorem}
\newcommand{\cV}{{\cal V}}
\newcommand{\cX}{{\cal X}}
\newcommand{\cY}{{\cal Y}}
\newcommand{\cZ}{{\cal Z}}
\begin{proof}
  For the proof of this theorem, we will use reduce from a hard to
  approximate constraint satisfaction problem (CSP) underlying the
  adaptive 3-query PCP given in \cite{GLST}. This PCP has perfect
  completeness and soundness $1/2+\eps$ for any desired constant
  $\eps$ (which is the best possible for $3$-query PCPs).

  We first state the properties of the CSP. An instance of the CSP
  will have variables partitioned into three parts $\cX,\cY$ and
  $\cZ$. Each constraint will be of the form $(x_i \vee (Y_j = z_k))
  \wedge (\overline{x_i} \vee (Y_j = z_l))$, where $x_i \in \cX$,
  $z_k,z_l \in \cZ$ are variables (unnegated) and $Y_j$ is a literal
  ($Y_j \in \{y_j, \overline{y_j}\}$ for some variable $y_j \in
  \cY$). 
For {\sc Yes} instances of the
  CSP, there will be a Boolean assignment that satisfies {\bf all} the
  constraints. For {\sc No} instances, every assignment to the
  variables will satisfy at most a fraction $(1/2+\eps)$ of the
  constraints. 

\begin{remark}
  We remark the condition that the instance is tripartite, and that the
  variables in $\cZ$ never appear negated are not explicit in
  \cite{GLST}. But these can be ensured by an easy modification to the
  PCP construction in \cite{GLST}. The PCP in \cite{GLST} has a
  bipartite structure: the proof is partitioned into two parts called
  the $A$-tables and $B$-tables, and each test consists of probing one
  bit $A(f)$ from an $A$ table and 3 bits $B(g),B(g_1),B(g_2)$ from
  the $B$ table, and checking $(A(f) \vee (B(g) = B(g_1)) \wedge
  (\overline{A(f)} \vee (B(g) = B(g_2))$. Further these tables are
  {\em folded} which is a technical condition that corresponds to the
  occurrence of negations in the CSP world. If the queries at
  locations $g_1$ and $g_2$ are made in a parallel $C$-table, and even
  if the $C$-table is not folded (though the $A$ and $B$ tables need
  to be folded), one can verify that the analysis of the PCP
  construction still goes through. This then translates to a CSP with
  the properties claimed above.
\end{remark}  

\begin{figure}[ht]
  \centering
  \begin{minipage}[b]{0.44\linewidth}
    \centering
    \includegraphics[scale=1]{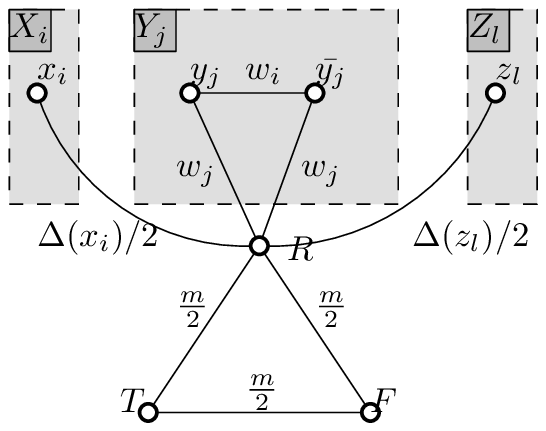}
    \caption{Global gadget for truth value assignments. Blocks $X_i$,
      $Y_j$ and $Z_l$ are replicated for all vertices in
      $\mathcal{X}$, $\mathcal{Y}$ and $\mathcal{Z}$. Edge weights are
      shown next to each edge.}
    \label{fig:gadget-global}  
  \end{minipage}
  \hspace{0.05\linewidth}
  \begin{minipage}[b]{0.47\linewidth}
    \centering
    \includegraphics[scale=1]{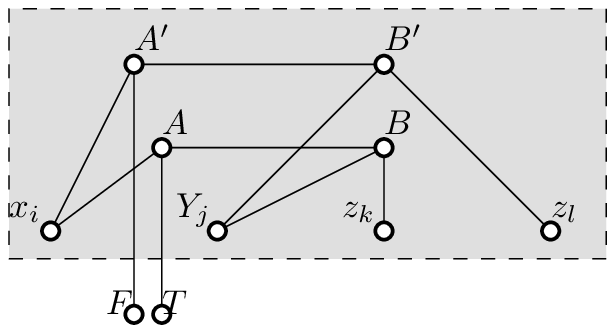}
    \caption{Local gadget for each constraint of the form $(x_i \vee
      Y_j = z_k) \wedge (\overline{x_i} \vee Y_j = z_l)$. All edges
      have unit weight. Labels $A,A',B,B'$ refer to the local nodes in
      each gadget.}
    \label{fig:gadget-local}
  \end{minipage}  
\end{figure}

Let ${\cal I}$ be an instance of such a CSP with $m$ constraints of
the above form on variables $\cV = \cX \cup \cY \cup \cZ$. Let
$\cX=\{x_1,x_2,\dots,x_{n_1}\}$, $\cY = \{y_1,y_2,\dots,y_{n_2}\}$ and
$\cZ=\{z_1,z_2,\dots,z_{n_3}\}$. From the instance ${\cal I}$ we
create a graph $G$ for the \maxcolor{3} problem as follows.  There is
a node $x_i$ for each variable $x_i \in \cX$, a node $z_l$ for each
$z_l \in \cZ$, and a pair of nodes $\{y_j,\overline{y_j}\}$ for the
two literals corresponding to each $y_j \in \cY$. There are also three
global nodes $\{R,T,F\}$ representing boolean values which are
connected in a triangle with edge weights $m/2$ (see
Fig. \ref{fig:gadget-global}).

For each constraint of the CSP, we place the local gadget specific to
that constraint shown in Figure \ref{fig:gadget-local}. Note that
there are 10 edges of unit weight in this gadget. The nodes
$y_j,\overline{y_j}$ are connected to node $R$ by a triangle whose
edge weights equal $w_j =
\frac{\Delta(y_j)+\Delta(\overline{y_j})}{2}$. Here $\Delta(X)$
denotes the total number of edges going from node $X$ \emph{into} all
the local gadgets. The nodes $x_i$ and $z_l$ connected to $R$ with an
edge of weight $\Delta(x_i)/2$ and $\Delta(z_l)/2$ respectively. The proofs of the following lemmas appear in Appendix~\ref{app:pf-3col}.
\begin{lemma}[Completeness]
\label{lem:3col-compl}
  Given an assignment of variables $\sigma: \cV \rightarrow \{0,1\}$ which
  satisfies at least $c$ of the constraints, we can construct a
  $3$-coloring of $G$ with at most $m-c$ improperly colored edges
  (each of weight 1).
\end{lemma}
%
\begin{lemma}[Soundness]
\label{lem:3col-soundness}
  Given a 3-coloring of $G$, $\chi$, such that the total weight of
  edges that are not properly colored by $\chi$ is at most $\tau <
  m/2$, we can construct an assignment $\sigma' : \cV \to\{0,1\}$ to
  the variables of the CSP instance that satisfies at least $m-\tau$
  constraints.
\end{lemma}

Returning to the proof of Theorem~\ref{thm:3color}, the total weight
of edges in $G$ is 
\[ 10 m + \frac{3 m}{2} + \underbrace{\sum_{i=1}^{n_1}
  \frac{\Delta(x_i)}{2}}_{m} + \sum_{j=1}^{n_2} 3 w_j +
\underbrace{\sum_{l=1}^{n_3} \frac{\Delta(z_l)}{2}}_{m} = \frac{27}{2}
m + \frac{3}{2} \underbrace{\sum_{j=1}^{n_2} (\Delta(y_i) +
  \Delta(\overline{y_j}))}_{2m} = \frac{33}{2} m \ . \]

  By the completeness lemma, {\sc Yes} instances of the CSP are mapped
  to graphs $G$ that are $3$-colorable. By the soundness lemma, {\sc
    No} instances of the CSP are mapped to graphs $G$ such that every
  $3$-coloring miscolors at least a fraction $\frac{(1/2-\eps)}{33/2}
  = \frac{1-2\eps}{33}$ of the total weight of edges. Since $\eps > 0$ is an arbitrary constant, the proof of Theorem~\ref{thm:3color} is complete.\footnote{Our reduction produced a graph with edge weights, but by Remark~\ref{rem:wt-unwt}, the same inapproximability factor holds for unweighted graphs as well.} 
\end{proof}
\subsection{\maxcolor{k} Hardness}
%
\begin{theorem}
\label{thm:kcolor-NP-hard}
  For every integer $k \ge 3$
  and every $\eps > 0$, it is NP-hard to
  approximate \maxcolor{k} within a factor of $1 - \frac{1}{33 (k + c_k)+c_k} +
  \eps$ where $c_k = k \mod 3 \le 2$.
\end{theorem}
\begin{proof}
  We will reduce \maxcolor{3} to \maxcolor{k} and then apply
  Theorem~\ref{thm:3color}. Throughout the proof, we will assume $k$
  is divisible by $3$. At the end, we will cover the remaining cases
  also. The reduction is inspired by the reduction from \mcut\ to
  \maxcut{k} given by Kann {\it et al.}~\cite{KKLP} (see
  Remark~\ref{rem:KKLP}). Some modifications to the reduction are
  needed when we reduce from \maxcolor{3}, and the analysis has to
  handle some new difficulties. The details of the reduction and its
  analysis follow.

  Let $G=(V,E)$ be an instance of \maxcolor{3}. By
  Theorem~\ref{thm:3color}, it is NP-hard to tell if $G$ is
  $3$-colorable or every $3$-colors miscolors a fraction $\frac{1}{33}
  -\eps$ of edges. We will construct a graph $H$ such that $H$ is
  $k$-colorable when $G$ is $3$-colorable, and a $k$-coloring which
  miscolors at most a fraction $\mu$ of the total weight of edges of
  $H$ implies a $3$-coloring of $G$ with at most a fraction $\mu k$ of
  miscolored edges. Combined with Theorem~\ref{thm:3color}, this gives
  us the claimed hardness of \maxcolor{k}.

Let $K'_{k/3}$ denote the complete graph with loops on $k/3$
vertices. Let $G'$ be the tensor product graph between $K_{k/3}$ and
$G$, $G'=K'_{k/3}\otimes G$ as defined by
Weichsel~\cite{weichsel62}. 
Identify each node in $G'$ with $(u,i),
u\in V(G), i\in \{1,2,\dots,k/3\}$. The edges of $G'$ are $((u,i),(v,i'))$ for $(u,v) \in E$ and any $i,i' \in \{1,\dots,k/3\}$. Next we make $3$ copies of $G'$, and identify
the nodes with $(u,i,j), (u,i)\in V(G'), j\in \{1,2,3\}$, then put edges
between all nodes of the form $(u,i,j)$ and $(u,i',j')$ if either
$i\neq i'$ or $j\neq j'$ with weight $\frac{2}{3} d_u$, where $d_u$ is
degree of node $u$. The total weight of edges in this new construction $H$ equals
\[ \sum_{u\in V} \left({k\choose 2} \frac{2}{3} d_u + \frac{3}{2}
  \left(\frac{k}{3}\right)^2 d_u \right) \le k^2 m \ . \]

\begin{lemma}
  If $G$ is $3$-colorable, then $H$ is $k$-colorable.
\end{lemma}
\begin{proof}
Let $\chi_G:V(G)\to \{1,2,3\}$ be a 3-coloring of $G$. Consider the
following coloring function for $H$, $\chi_H:V(H)\to \{1,2,\dots,k\}$. For node
$(u,i,j)$, let $\chi_H((u,i,j)) = \pi^{j}(\chi_G(u)) + 3 (i-1)$. Here
$\pi$ is the permutation $\left(\begin{array}{ccc} 1 & 2 & 3
  \\ 2&3&1 \end{array}\right)$, and $\pi^j(x) =
\underbrace{\pi(\ldots(\pi(}_{\text{j times}} x)))$. Equivalently
$\pi(x) = x\mod 3 + 1$.

Consider edges of the form $\{(u,i,j),(v,i',j)\}$. If $i\neq i'$, then
colors of the endpoints are different. Else we have
$\chi((u,i,j))-\chi((v,i,j)) \equiv \chi(u)-\chi(v) \not\equiv 0 \mod
3$. For edges of the form $\{(u,i,j),(u,i',j')\}$, if $i\neq i'$,
clearly edge is satisfied. When $i=i',j\neq j'$, $\chi((u,i,j))-\chi((u,i,j'))
\equiv \pi^{j}(u) - \pi^{j'}(u) \equiv j-j' \not\equiv 0 \mod 3$.
\end{proof}

\begin{lemma}
  If $H$ has a $k$-coloring that properly colors a set of edges with
  at least a fraction $(1-\mu)$ of the total weight, then $G$
  has a $3$-coloring which colors at least a fraction $(1-\mu k)$ of
  its edges properly.
\end{lemma}
\begin{proof}
Let $\chi_H$ be the coloring of $H$, $\sugg^j_u = \{\chi_H((u,i,j)) \mid 1 \le i\le k/3\}$ and $\sugg_u = \bigcup_j \sugg^j_u$. Denote the total weight of uncut
edges in this solution as 
\begin{equation}
C^{total} = \sum_{u \in V(G)} \frac{2}{3} d_u C^{within}_u +
C^{between},  
\end{equation}
where $C^{within}_u$ and $C^{between}$ denotes the number of
improperly colored edges within the copies of node $u$ and between
copies of different vertices $u,v \in V(G)$ respectively. We have the
following relations:
\begin{equation}
\label{eq:c-between}
  \begin{array}{rclr}
  C^{between} &=& \sum_{j=1}^3 \sum_{uv \in E(G)} \sum_{1\le i\le
    i'\le k/3} 1_{\chi_H((u,i,j))=\chi_H((v,i',j))} \\ &\ge&
  \sum_{j=1}^3 \sum_{uv \in E(G)} |\sugg^j_u \cap \sugg^j_v|
  \end{array}
\end{equation}

\begin{equation}
  \begin{array}{rclr}
    C^{within}_u &=& \sum_{c \in \sugg_u} {|\chi_H^{-1}(c) \cap B_u|
      \choose 2} & \mbox{($B_u=\{(u,i,j)|\forall i,j\}$)} \\ &=&
    \sum_{c \in \sugg_u} \frac{|B_{u,c}|^2}{2} -
    \frac{k}{2} & \mbox{($B_{u,c} = B_u \cap \chi_H^{-1}(c)$)}\\
    &\ge& \frac{1}{2 |\sugg_u|} \left(\sum_{c \in \sugg_u}
    |B_{u,c}|\right)^2  - \frac{k}{2} & \mbox{(Cauchy-Schwarz)}\\
    &\ge& \frac{k}{2}\left(\frac{k}{|\sugg_u|}-1\right) 
      \ge \frac{k}{2} \frac{|\overline{\sugg_u}|}{|\sugg_u|}
      \ge \frac{|\overline{\sugg_u}|}{2}
  \end{array}
\end{equation}

%
%
%
%
Now we will find a (random) $3$-coloring $\chi_G$ for $G$.  Pick $c$ from $\{1,2,\dots,k\}$
uniformly at random.  If $c \notin \sugg_u$, select $\chi_G(u)$
uniformly at random from $\{1,2,3\}$. If $c \in \sugg_u$, set
$\chi_G(u) = j$ if $j$ is the smallest index for which $c \in
\sugg^j(u)$.  With this coloring $\chi_G(u)$, the probability that an
edge $(u,v) \in E(G)$ will be improperly colored is:
\begin{eqnarray*}
\prob{}{\chi_G(u) = \chi_G(v)} 
& \le & \sum_{j=1}^3 \prob{c}{c \in \sugg^j_u \cap \sugg^j_v} + \frac{1}{3} \prob{c}{c \in \overline{\sugg_u},c\in \sugg_v} \\
& & \quad + ~~
\frac{1}{3} \prob{c}{c \in \sugg_u,c\in \overline{\sugg_v}}
+ \frac{1}{3} \prob{c}{c \in \overline{\sugg_u},c\in \overline{\sugg_v}}\\
&\le& \sum_{j=1}^3 \frac{|\sugg^j_u\cap \sugg^j_v|}{k} + 
\frac{|\overline{\sugg_u}|}{3k} + \frac{|\overline{\sugg_v}|}{3k} 
\end{eqnarray*}
%
\noindent We can thus bound the expected number of
miscolored edges in the coloring $\chi_G$ as follows.
\begin{eqnarray*}
\mathbb{E} \Biggl[ \sum_{(u,v) \in E(G)} 1_{\chi_G(u)=\chi_G(v)}\Biggr] &\le& \sum_{uv \in E} \Biggl[  \biggl( \sum_{j=1}^3 \frac{|\sugg^j_u\cap
      \sugg^j_v|}{k} \biggr)+ \frac{|\overline{\sugg_u}|}{3k} +
    \frac{|\overline{\sugg_v}|}{3k} \Biggr] \\ 
&\le& \frac{1}{k} \Bigl( C^{between} +
\sum_{u\in V(G)} \frac{d_u}{3} |\overline{\sugg_u}|\Bigr) \quad \mbox{(using (\ref{eq:c-between}))} \\
&\le& \frac{1}{k} \Bigl(C^{between} + 
      \sum_{u\in V(G)} \frac{2 d_u}{3}C^{within}_u  \Bigr) = \frac{C^{total}}{k}
\end{eqnarray*}
This implies that there exists a $3$-coloring of $G$ for which the
number of improperly colored edges in $G$ is at most
$\frac{C^{total}}{k}$. Therefore if $H$ has a $k$-coloring which
improperly colors at most a total weight $\mu k^2 m$ of edges, then
there is a $3$-coloring of $G$ which colors improperly at most a fraction
$\frac{\mu k^2 m}{k m} = \mu k$ of its edges.
\end{proof}

This completes the proof of Theorem~\ref{thm:kcolor-NP-hard} when $k$ is divisible by $3$. The other cases are easily handled by adding $k \mod 3$ extra nodes connected to all vertices by edges of suitable weight. See Appendix~\ref{app:k-mod-3} for details. 
\end{proof}
\begin{remark}[Comparison to \cite{KKLP}]
\label{rem:KKLP}
The reduction of Kann {\it et al}~\cite{KKLP} converts
  an instance $G$ of \mcut\ to the instance $G' = K'_{k/2} \otimes G$
  of \maxcut{k}. Edge weights are picked so that the optimal $k$-cut
  of $G'$ will give a set $S_u$ of $k/2$ different colors to all
  vertices in each $k/2$ clique $(u,i)$, $1 \le i \le k/2$. This
  enables converting a $k$-cut of $G'$ into a cut of $G$ based on
  whether a random color falls in $S_u$ or not. In the $3$-coloring
  case, we make $3$ copies of $G'$ in an attempt to enforce three ``translates''
  of $S_u$, and use those to define a $3$-coloring from a
  $k$-coloring. But we cannot ensure that each $k$-clique is properly
  colored, so these translates might overlap and a more careful
  soundness analysis is needed.
\end{remark}
\section{Conditional Hardness Results for \maxcolor{k}}
We will first review the (exact) \dtoone{2} Conjecture, and then
construct a noise operator, which allows us to preserve
$k$-colorability. Then we will bound the stability of coloring
functions with respect to this noise operator. In the last section, we
will give a PCP verifier which concludes the hardness result.
\vspace{-1ex}
\subsection{Preliminaries}
We begin by reviewing some definitions
and \dtoone{d} conjecture.
\begin{definition}
  An instance of a bipartite Label Cover problem represented as
  $\mathcal{L} = (U, V, E, W, R_U, R_V, \Pi)$ consists of a weighted
  bipartite graph over node sets $U$ and $V$ with edges $e=(u,v) \in
  E$ of non-negative real weight $w_e \in W$. $R_U$ and $R_V$ are
  integers with $1\le R_U\le R_V$. $\Pi$ is a collection of projection
  functions for each edge:
  $\Pi=\{\pi_{vu}:\intset{R_V}\rightarrow\intset{R_U}\big|u\in U, v\in
  V\}$. A labeling $\ell$ is a mapping $\ell:U\rightarrow
  \intset{R_U}$, $\ell:V\rightarrow \intset{R_V}$. An edge $e=(u,v)$ is
  satisfied by labeling $\ell$ if $\pi_e(\ell(v))=\ell(u)$. We define the
  value of a labeling as sum of weights of edges satisfied by this
  labeling normalized by the total weight. $\opt(\mathcal{L})$ is the
  maximum value over any labeling.
\end{definition}
\begin{definition}
  A projection $\pi:\intset{R_V}\rightarrow \intset{R_U}$ is called
  \dtoone{d} if for each $i \in \intset{R_U}$, $|\pi^{-1}(i)|\le
  d$. It is called {\em exactly} \dtoone{d} if $|\pi^{-1}(i)|= d$ for
  each $i \in \{1,2,\dots,R_U\}$.
\end{definition}
\begin{definition}
  A bipartite Label-Cover instance $\mathcal{L}$ is called $d$-to-$1$
  Label-Cover if all projection functions, $\pi \in \Pi$ are
  $d$-to-$1$.
\end{definition}
\begin{conjecture}[\dtoone{d} Conjecture \cite{Khot02}]
  For any $\gamma > 0$, there exists a \dtoone{d} Label-Cover instance
  $\mathcal{L}$ with $R_V=R(\gamma)$ and $R_U \le d R_V$ many labels
  such that it is NP-hard to decide between two cases,
  $\opt(\mathcal{L}) = 1$ or $\opt(\mathcal{L})\ge \gamma$.  Note that
  although the original conjecture involves \dtoone{d} projection
  functions, we will assume that it also holds for \emph{exactly}
  \dtoone{d} functions (so $R_U = d R_V$), which is the case in
  \cite{DMR}.
\label{conj:dtoone}
\end{conjecture}
Using the reductions from \cite{DMR}, it is possible to show that the
above conjecture still holds given that the graph $(U\cup V, E)$ is
left-regular and unweighted, i.e., $w_e = 1$ for all $e\in E$.
\subsection{Noise  Operators}
For a positive integer $M$, we will denote by $[M]$ the set
$\{0,1,\dots,M-1\}$.  We will identify elements of $[M^2]$ with $[M]
\times [M]$ in the obvious way, with the pair $(a,b) \in [M]^2$
corresponding $a + M b \in [M^2]$. 
\begin{definition}
  \label{def:markovop}
  A Markov operator $T$ is a linear operator which maps probability
  measures to other probability measures. In a finite discrete
  setting, it is defined by a stochastic matrix whose $(x,y)$'th entry
  $T(x \rightarrow y)$ is the probability of transitioning from $x$ to
  $y$. Such an operator is called symmetric if $T(x \rightarrow y) =
  T(y \rightarrow x) = T(\xiffy{x}{y})$.
\end{definition}
\begin{definition}
\label{def:beckner}
  Given $\rho\in [-1,1]$, the Beckner noise operator, $T_{\rho}$ on
  $[q]$ is defined by as $T_\rho(x\rightarrow x) = \frac{1}{q} +
  \left(1-\frac{1}{q}\right) \rho$ and $T_\rho(x\rightarrow y) =
  \frac{1}{q}(1-\rho)$ for any $x\neq y$. 
\end{definition}
\begin{observation}
  All eigenvalues of the operator $T_{\rho}$ are given by $1 =
  \lambda_0(T_\rho)\ge \lambda_1(T_\rho) = \ldots =
  \lambda_{q-1}(T_\rho) = \rho$. Any orthonormal basis
  $\alpha_0,\alpha_1,\ldots,\alpha_{q-1}$ with $\alpha_0$ being
  constant vector, is also a basis for $T_{\rho}$.
\end{observation}
\begin{lemma}
  \label{lemma:T}
For an integer $q \ge 6$, there exists a symmetric Markov operator $T$ on $[q]^2$ whose diagonal entries are all $0$ and with 
 eigenvalues  $1=\lambda_0\ge
  \lambda_1\ge \ldots\ge \lambda_{q^2-1}$ such that the {\em spectral radius}
  $\rho(T) = \max\{|\lambda_1|,|\lambda_{q^2-1}|\}$ is at most
  $\frac{4}{q-1}$.
\end{lemma}
\begin{proof}
  Consider the symmetric Markov operator $T$ on $[q]^2$ such that, for
  $x=(x_1,x_2), y=(y_1,y_2)\in [q]^2$,
  \[T(\xiffy{x}{y}) =
  \begin{cases}
    \alpha &\mbox{if $\{x_1,x_2\}\cap\{y_1,y_2\}=\emptyset$
      and $x_1\neq x_2, y_1\neq y_2$,}\\
    \beta &\mbox{if $x_1\not\in\{y_1,y_2\}$
      and $x_1= x_2, y_1\neq y_2$,}\\
    \beta &\mbox{if $y_1\not\in\{x_1,x_2\}$
      and $x_1\neq x_2, y_1= y_2$,}\\
    0 & \mbox{else,}
  \end{cases}\] where $\alpha = \frac{1}{(q-1)(q-3)}$ and 
  $\beta = \frac{1}{(q-1)(q-2)}$. It is clear that $T$ is symmetric and
  doubly stochastic.

  To bound the spectral radius of $T$, we will bound the second
  largest eigenvalue $\lambda_1(T^2)$ of $T^2$. Notice that $T^2$ is
  also a symmetric Markov operator. Moreover $\lambda_i(T^2) =
  \lambda^2_i(T)$, therefore $\lambda_1(T^2) \ge
  \max(\lambda^2_1(T),\lambda^2_{q^2-1}(T)) \ge \rho(T)^2$.

  Notice that $T^2(\xiffy{x}{y})>0$ for all pairs $x,y\in[q]^2$. 
  Consider the variational
  characterization of $1 - \lambda_1(T^2)$ \cite{Sinclair92}:
  \[\min_\psi \frac{\sum_{x,y}
    (\psi(x) - \psi(y))^2 \pi(x) T^2(\xiffy{x}{y})}{\sum_{x,y}
    (\psi(x)-\psi(y))^2 \pi(x)\pi(y)} \ge \min_\psi \min_{x,y}
  \frac{\pi(x) (\psi(x) - \psi(y))^2
    T^2(\xiffy{x}{y})}{(\psi(x)-\psi(y))^2 \pi(x)\pi(y)} = \min_{x,y}
  q^2 T^2(\xiffy{x}{y})\]


  For any two pairs $(x_1,x_2),(y_1,y_2)\in[q]^2$, let $l =
  |[q]\setminus \{x_1,x_2,y_1,y_2\}|$. Then we have
  \begin{eqnarray*}
    T^2(\xiffy{(x_1,x_2)}{(y_1,y_2)}) &=& 
  \begin{cases} l (l-1)\beta^2\ge (q-2)(q-3)\beta^2 & 
    \mbox{if $x_1=x_2$ and $y_1=y_2$,}\\
    l (l-1) \alpha \beta\ge (q-3)(q-4)\alpha\beta & 
    \mbox{if $x_1\neq x_2$ and $y_1=y_2$,}\\
    l (l-1) \alpha \beta\ge (q-3)(q-4)\alpha\beta & 
    \mbox{if $x_1=x_2$ and $y_1\neq y_2$,}\\
    l (l-1) \alpha^2 + l \beta^2 \ge (q-4)(q-5)\alpha^2 + (q-4)\beta^2 &
    \mbox{if $x_1\neq x_2$ and $y_1\neq y_2$.}
  \end{cases}\\
  &\ge& \frac{(q-5) (q-4)}{(q-3)^2 (q-2) (q-1)}
  \end{eqnarray*}
  So $\rho(T)\le \sqrt{\lambda_1(T^2)} \le \sqrt{1-\frac{(q-5) (q-4)
      q^2}{(q-3)^2 (q-2) (q-1)}} \le \frac{3}{q} + 
      \frac{8}{q^2}\le\frac{4}{q-1}$ for $q\ge 6$.
\end{proof}
\subsection{$q$-ary Functions, Influences, Noise stability}
We define inner product on space of functions from $[q]^N$ to
$\R$ as $\langle f,g \rangle = \expct{x\sim[q]^N}{f(x)
  g(x)}$. Here $x \sim \mathcal{D}$ denotes sampling from distribution
$\mathcal{D}$ and $\mathcal{D} = [q]^N$ denotes the uniform
distribution on $[q]^N$.

Given a symmetric Markov operator $T$ and
$x=(x_1,\ldots,x_N)\in[q]^N$, let $T^{\otimes N} x$ denote the product
distribution on $[q]^N$ whose \texsup{$i$}{th} entry $y_i$ is
distributed according to $T(\xiffy{x_i}{y_i})$. Therefore $T^{\otimes
  N} f (x) = \expct{y \sim T^{\otimes N} x} { f(y)}$.

\begin{definition}
  Let $\alpha_0,\alpha_1,\ldots,\alpha_{q-1}$ be an orthonormal basis
  of $\R^q$ such that $\alpha_0$ is all constant vector. For
  $x\in[q]^N$, we define $\alpha_x \in \R^{q^N}$ as
  \[ \alpha_x = \alpha_{x_1}\otimes \ldots \otimes \alpha_{x_N}.\]
\end{definition}

\begin{definition}[Fourier coefficients]
  For a function $f:[q]^N\to\R$, define $\hat{f}(\alpha_x) =
  \langle f, \alpha_x \rangle$.
\end{definition}
\begin{definition}
  Let $f:[q]^N\rightarrow \R$ be a function. The
  \emph{influence} of \texsup{i}{th} variable on $f$, $\infl_i(f)$ is
  defined by
  \[ \infl_i(f) = 
      \expct{}
            {\var{}{f(x)|x_1,\ldots,x_{i-1},x_{i+1},\ldots,x_N}}\]
  where $x_1,\ldots,x_N$ are uniformly distributed.
  Equivalently, $\infl_i (f) = \sum_{x:x_i\neq 0} \hat{f}^2(\alpha_x)$.
\end{definition}
\begin{definition}
  Let $f:[q]^N\rightarrow \R$ be a function. The
  \emph{low-level influence} of \texsup{i}{th} variable of $f$ is
  defined by
  \[ \infl_i^{\le t}(f) = \sum_{x:x_i\neq 0,\ |x|\le t} \hat{f}^2(\alpha_x).\]
\end{definition}
\begin{observation}
  For any function $f$, $\sum_i \infl_i^{\le t}(f) = \sum_{x:|x|\le t}
  \hat{f}^2(\alpha_x) |x|\le t \sum_x \hat{f}^2(\alpha_x) = t
  \|f\|^2_2$. If $f:[q]^N\to[0,1]$, then $\|f\|^2_2 \le 1$, so $\sum_i
  \infl_i^{\le t}(f) \le t$.
\end{observation}
\begin{definition}[Noise stability]
  Let $f$ be a function from $[q]^N$ to $\R$, and let $-1\le
  \rho\le 1$. Define the \emph{noise stability} of $f$ at $\rho$ as
  \[ \mathbb{S}_{\rho}(f) = \langle f, T^{\otimes n}_\rho f\rangle =
  \sum_x \rho^{|x|} \hat{f}^2_i(\alpha_x)\] where $T_{\rho}$ is the
  Beckner operator as in Definition~\ref{def:beckner}.
\end{definition}
A natural way to think about a $q$-coloring function is as a
collection of $q$-indicator variables summing to $1$ at every point.
To make this formal:
\begin{definition}
  Define the unit $q$-simplex as $\Delta_q = \{(x_1,\ldots,x_{q}) \in
  \R^q \mid \sum x_i = 1, x_i \ge 0\}$.
\end{definition}
\begin{observation}\label{obs:infbndsimplex}
  For positive integers $Q,q$ and any function
  $f=(f_1,\ldots,f_q):[Q]^N\to \Delta_q$, $\sum_i \infl^{\le t}_i(f) =
  \sum_i \sum_j \infl^{\le t}_i (f_j) \le t \sum_j \|f_j\|^2 \le t$.
\end{observation}
We want to prove a lower bound on the stability of $q$-ary functions
with noise operators $T$. The following proposition is generalization
of Proposition 11.4 in \cite{KKMO} to general symmetric Markov
operators $T$ with small spectral radii. Its proof appears in Appendix~\ref{app:noise-stab}.
\begin{proposition}
\label{prop:stab}
For integers $Q,q \ge 3$, and a symmetric Markov operator $T$ on $[Q]$
with spectral radius $\rho(T)\le \frac{c}{q-1}$, for some $c>0$, there
is a small enough $\delta=\delta(q)>0$ and $t=t(q)>0$ such that for
any function $f=(f_1,\ldots,f_q):[Q]^N\rightarrow \Delta_q$ with
$\infl_i^{\le t} (f)\le \delta$, for all $i$, satisfies
  \[ \sum_{j=1}^q \langle f_j, T^{\otimes N} f_j\rangle \ge 1/q - 2 c \ln q/q^2 - C
  \ln\ln q /q^2\] for some universal constant $C<\infty$.
\end{proposition}
\begin{definition}[Moving between domains]
  For any $x=(x_1,\ldots,x_{2 N}) \in [q]^{2 N}$, denote $\overline{x}
  \in [q^2]^N$ as 
\[ \overline{x} = ((x_1, x_2),\ldots,(x_{2 N-1},x_{2 N})) \ . \]
  Similarly for $y=(y_1,\ldots,y_{N}) \in [q^2]^{N}$, denote
  $\underline{y} \in [q]^{2 N}$ as
  \[ \underline{y} = (y_{1,1}, y_{1,2},\ldots,y_{N,1},y_{N,2}),\]
  where $y_i = y_{i,1} + y_{i,2} q$ such that $y_{i,1},y_{i,2}\in
  [q]$.  For a function $f$ on $[q]^{2 N}$, define $\overline{f}$ on
  $[q^2]^N$ as $\overline{f}(y) = f(\underline{y})$. 
\end{definition}
\noindent 
The relationship between influences of variables for functions $f$ and
$\overline{f}$ are given by the following claim (Claim 2.7 in
\cite{DMR}).
\begin{claim}\label{claim:infrel}
  For any function $f:[q]^{2N}\rightarrow \R$,
  $i\in\intset{N}$ and any $t\ge 1$,
  $\infl_i^{\le t}(\overline{f}) \le \infl_{2i-1}^{\le 2t}(f) +
  \infl_{2i}^{\le 2t} (f)$.
\end{claim}

%
%
%
%
%
%
%
%

\vspace{-1ex}
\subsection{PCP Verifier for \maxcolor{k}}
This verifier uses ideas similar to the \maxcut{k} verifier given in
\cite{KKMO} and the $4$-coloring hardness reduction in \cite{DMR}. Let $\mathcal{L} = (U,V,E,R,2 R,\Pi)$ be a \dtoone{2}
bipartite, unweighted and left regular Label-Cover instance as in
Conjecture \ref{conj:dtoone}. Assume the proof is given as the Long
Code over $[k]^{2 R}$ of the label of every vertex $v \in V$. Below
for a permutation $\sigma$ on $\intset{n}$ and a vector $x \in \R^n$, $x\circ \sigma$ denotes $(x_{\sigma(1)},x_{\sigma(2)},\cdots,x_{\sigma(n)})$. For a function $f$ on $\R^n$, $f \circ \sigma$ is defined as $f \circ \sigma(x) = f(x \circ \sigma)$.



\begin{itemize}
\vspace{-1ex}
\itemsep=0ex
\item Pick $u$ uniformly at random from $U$, $u\sim U$.
\item Pick $v,v'$ uniformly at random from $u$'s neighbors. Let $\pi,
  \pi'$ be the associated projection functions, $\chi_v, \chi_{v'}$ be
  the (supposed) Long Codes for the labels of $v, v'$ respectively.
\item Let $T$ be the Markov operator on $[k]^2$ given in Lemma
  \ref{lemma:T}. Pick $x \sim [k^2]^R$ and $y \sim T^{\otimes R}
  x$. Let $\sigma_v,\sigma_{v'}$ be two permutations of $\intset{2 R}$ such
  that $\pi(\sigma^{-1}_v(2i-1)) = \pi(\sigma^{-1}_v(2i)) =
  \pi'(\sigma_{v'}^{-1}(2i-1)) = \pi'(\sigma_{v'}^{-1}(2i))$ (both $\pi$ and
  $\pi'$ are exactly \dtoone{2}, so such permutations exist).
\item Accept iff $\chi_v\ \circ\ \sigma_v (\underline{x}) $ and
  $\chi_{v'} \ \circ\ \sigma_{v'}(\underline{y})$ are different.
\end{itemize}
The proofs of the following two lemmas appear in Appendix~\ref{app:kcol-PCP}.
\begin{lemma}[Completeness]
\label{lem:kcol-compl}
  If the original \dtoone{2} Label-Cover instance $\mathcal{L}$ has a
  labeling which satisfies all constraints, then there is a proof
  which makes the above verifier always accept.
\end{lemma}
\begin{lemma}[Soundness]
\label{lem:kcol-soundness}
  There is a constant $C$ such that, if the above verifier passes with
  probability exceeding $1-1/k+O(\ln k/k^2)$, then there
  is a labeling of $\mathcal{L}$ which satisfies $\gamma'=\gamma'(k)$
  fraction of the constraints independent of label set size $R$.
\end{lemma}
Note that our PCP verifier makes ``$k$-coloring'' tests. By the standard conversion from PCP verifiers to CSP hardness, and Remark~\ref{rem:wt-unwt} about conversion to unweighted graphs with the same inapproximability factor, we conclude the main result of this section by combining Lemmas \ref{lem:kcol-compl} and \ref{lem:kcol-soundness}.
\begin{theorem}
  For any constant $k \ge 3$, assuming \dtoone{2} Conjecture, it is
  NP-hard to approximate \maxcolor{k} within a factor of $1-1/k+O(\ln
  k/k^2)$. 
\end{theorem}

\bibliographystyle{abbrv}
\bibliography{kcolor}

\appendix

\section{Proofs from Section~\ref{sec:3color-hard}}
\label{app:pf-3col}

\subsection{Proof of Lemma~\ref{lem:3col-compl}}
\label{app:pf-3col-compl}
\begin{proof}
  We define the coloring $\chi:V(G)\to [3]$ in the obvious way, with
  nodes $T$, $R$ and $F$ fixed to different colors. Then define 
\[\chi(x_i)
  = \begin{cases} \chi(T) &\mbox{if $\sigma(x_i)=1$,}\\ \chi(F)
    &\mbox{else.}\end{cases}\]
and similarly for the nodes $y_j$, $z_l$. Define
\[\chi(\overline{y_i})
  = \begin{cases} \chi(F) &\mbox{if $\sigma(y_j)=1$,}\\ \chi(T)
    &\mbox{else.}\end{cases}\]

  Now, for the constraints satisfied by this assignment, $(x_i \vee
  (Y_j = z_k)) \wedge (\overline{x_i} \vee (Y_j = z_l))$, consider the
  corresponding gadget. Let $\sugg(A) = [3]\setminus \{\chi(x_i),
  \chi(T)\}$ and $\sugg(B) = [3]\setminus \{\chi(Y_j), \chi(z_k)\}$ be
  the available colors to $A$ and $B$ which can properly color all
  edges incident to variables. Notice that none of these sets are
  empty and since $x_i \vee (Y_j = z_k)$ is true, at least one of
  these sets $\sugg(A)$ and $\sugg(B)$ has two elements in it. Hence
  there exists a coloring of $A$ and $B$ from sets $\sugg(A)$ and
  $\sugg(B)$ such that $\chi(A)\neq \chi(B)$. The same argument also
  holds for $A'$ and $B'$, therefore all edges in this gadget are
  properly colored.

  For the violated constraints, either $\sugg(A)$ or $\sugg(A')$ has
  one element. Augmenting that set with the color $\chi(x_i)$ will
  cause only one edge to be violated. 
\end{proof}

\subsection{Proof of Lemma~\ref{lem:3col-soundness}}

\begin{proof}  
   Since $\tau <
  m/2$, the coloring $\chi$ must give three different colors to the
  nodes $T$, $F$, and $R$.  If $\chi(x_i) = \chi(R)$, then randomly
  choosing $\chi(x_i)$ from $\{\chi(T), \chi(F)\}$ will, in
  expectation, make at most half of the local gadget edges going out
  of $x_i$ improperly colored, which is exactly the value
  $\Delta(x_i)/2$ gained. So we can assume that $\chi(x_i) \in
  \{\chi(T),\chi(F)\}$ for each $x_i$. A similar argument holds for
  the nodes $z_l$. Now consider the nodes $y_j$ and $\overline{y_j}$
  for a variable in $Y$. If $\chi(y_j) = \chi(R)$,
  $\chi(\overline{y_j})=\chi(R)$ or $\chi(x_j) =
  \chi(\overline{y_j})$, then randomly choosing
  $(\chi(y_j),\chi(\overline{y_j}))$ from
  $\{(\chi(T),\chi(F)),(\chi(F),\chi(T))\}$ will, in expectation, make
  at most half of the local gadget edges going out of nodes $y_j$ and
  $\overline{y_j}$ improperly colored, which is exactly the value
  $w_j$ gained.

  To summarize, we can assume that nodes $T$,$F$ and $R$ are colored
  differently, $\chi(x_i),\chi(Y_j),\chi(z_l) \in \{\chi(T),\chi(F)\}$
  and $\chi(y_j)\neq \chi(\overline{y_j})$. Thus all edges other than
  the edges inside the local gadgets are properly colored by
  $\chi$, and by assumption at most $\tau$ edges are miscolored by $\chi$.

  Now define the natural assignment $\sigma'$ that assigns a variable
  of $\cV$ the value $1$ if the associated variable received the color
  $\chi(T)$, and the value $0$ if its color is $\chi(F)$.


  Consider a local gadget, with all edges properly colored,
  corresponding to the constraint $(x_i \vee (Y_j = z_k)) \wedge
  (\overline{x_i} \vee (Y_j = z_l))$. Assume $\sigma'(x_i) = 0$, which
  implies $\chi(A) = \chi(R)$. Then both neighbors of $B$ besides $A$
  must have the same color, therefore $\sigma(Y_j) = \sigma(z_k)$. The
  other case when $\sigma'(x_i)=1$ is similar. Hence the assignment
  $\sigma'$ will satisfy this constraint.

Since the local gadgets corresponding to different constraints have disjoint sets of edges, it follows that the number of constraints violated by the assignment $\sigma'$ is at most $\tau$. 
\end{proof}

\section{Proof of Proposition~\ref{prop:stab}}
\label{app:noise-stab}

\begin{proof}
  Let $t=4$, $f_i:[Q]^N\rightarrow[0,1]$ denote the \texsup{$i$}{th}
  coordinate function of $f$, and let $\mu_i = \expct{}{f_i}$. Let
  $\alpha_0,\ldots, \alpha_{Q-1}$ be an orthonormal set of
  eigenvectors for $T$ with corresponding eigenvalues ${\lambda_0\ge
    \ldots\ge \lambda_{Q-1}}$, with $\rho=\rho(T)\le \frac{c}{q-1}$
  being the spectral radius of $T$. Notice that $T$ is symmetric so
  $\lambda_0 = 1$ and $\alpha_0$ is a constant vector. Therefore
  $\expct{}{f_i} = \hat{f}_i(\alpha_0) = \mu_i$. Then (using the
  notation from \cite{DMR}):
  \[ T^{\otimes N} \alpha_x = (\prod_{a \neq 0} \lambda_a ^{|x|_a}) \alpha_x \]
  and hence
  \[ T^{\otimes N} f_i = \sum_x (\prod_{a \neq 0} \lambda_a ^{|x|_a})
  \hat{f}_i(\alpha_x) \alpha_x. \]


  At this point, consider the Beckner operator, $T_{\rho}$ on
  $[Q]$. Since $\alpha_0$ is the uniform distribution, it is a
  constant vector, thus $\alpha_0,\alpha_1,\ldots,\alpha_{Q-1}$ is
  also an orthonormal basis for $T_{\rho}$. Consequently,
  \begin{eqnarray*}
    \langle f_i, T_{\rho} ^{\otimes N} f_i\rangle &=& 
    \sum_x (\prod_{a \neq 0} \rho ^{|x|_a})
    \hat{f}_i^2(\alpha_x) = \sum_x \rho^{|x|} \hat{f}^2_i(\alpha_x) = 
    \mathbb{S}_{\rho}(f_i)
  \end{eqnarray*}
  Thus
  \begin{eqnarray*}
    \langle f_i, T^{\otimes N} f_i\rangle &=& 
    \hat{f}_i^2(\alpha_0) - \hat{f}_i^2(\alpha_0) + 
      \sum_x \underbrace{(\prod_{a \neq 0} \lambda_a ^{|x|_a})}_
      {
        \begin{cases}
          \ge - \rho^{|x|} & \text{if }|x|\neq 0,  \\
          = 1 & \text{else.}
        \end{cases}
      }
      \hat{f}_i^2(\alpha_x)\\
      &\ge& 2 \mu_i^2 - 
      \sum_{x} \rho^{|x|} \hat{f}_i^2(\alpha_x) =
      2\mu_i^2 -\sum_{x:|x|\le 4} \rho^{|x|} \hat{f}_i^2(\alpha_x) 
      -\sum_{x:|x|> 4} \rho^{|x|} \hat{f}_i^2(\alpha_x)\\
      &\ge& 2\mu_i^2 
      -\sum_{x:|x|\le 4} \rho^{|x|} \hat{f}_i^2(\alpha_x) 
      - \rho^{4}\\
      &\ge&  2\mu_i^2 
      -\sum_{x:|x|\le 4} \rho^{|x|} \hat{f}_i^2(\alpha_x) 
      - q^{-3}\\
  \end{eqnarray*}
  At this point, let $\tilde{f}_i(x) = \sum_{x:|x|\le 4} (\prod_{a
    \neq 0} \lambda_a ^{|x|_a}) \hat{f}_i(\alpha_x)\alpha_x $ be the
  function having the same low-level coefficients with $f_i(x)$ and
  $0$ for the higher-levels. It is easy to verify that
  $\expct{}{\tilde{f}_i} = \mu_i$, $\infl_i(f_j) \ge
  \infl_i(\tilde{f}_j) = \infl_i^{\le 4}(f_j)$ and
  $\mathbb{S}_{\rho}(\tilde{f_j}) = \sum_{x:|x|\le 4} \rho^{|x|}
  \hat{f}_i^2(\alpha_x)$. In particular, our assumption $\sum_j
  \infl_i^{\le t}(f_j)=\sum_j \infl_i^{\le 4}(f_j) \le \delta$ implies
  $\sum_j \infl_i(\tilde{f}_j) \le \delta$.

  Let $\delta$ be a small enough constant such that
  $\mathbb{S}_{\frac{c}{q-1}}(\tilde{f}_i) \le
  \Gamma_{\frac{c}{q-1}}(\mu_i)+\eps$ for some small $\eps\le
  \frac{1}{q^3}$, from the Majority is Stablest Theorem \cite{MOO}. In
  \cite{KKMO}, $\Lambda_{\eta}(\mu)$ is used for $\Gamma_{\eta}(\mu)$
  and we will follow that convention instead. Below, for a real $x$, $[x]^+$ denotes $\max \{ x, 0\}$. Then
  \begin{eqnarray*}
    \sum_i \langle f_i, T^{\otimes N} f_i\rangle &\ge& \sum_i \left[2 \mu_i^2
    - \mathbb{S}_{\rho}(\tilde{f}_i)\right]  - q^{-2}\\
    &\ge& \sum_i \left[2 \mu_i^2 - \mathbb{S}_{\frac{c}{q-1}}(\tilde{f}_i)\right]  -
    q^{-2}\\
    &\ge& \sum_i \left[2 \mu_i^2 - \Lambda_{\frac{c}{q-1}}(\mu_i)\right]^+-2 q^{-2}\\
    &\ge& \frac{1}{q} - \frac{2 c \ln q}{q^2} - O\left(\frac{\ln\ln q}{q^2}\right)
  \end{eqnarray*}
  The last inequality is proved in the same way as Proposition 11.4 in
  \cite{KKMO}. The only difference is that we have
  \[ F(\mu_i) = \mu_i^2 + \frac{c}{q-1} 2 \mu_i^2 \ln(1/\mu_i)\cdot
  \left(1+C\frac{\ln\ln q}{\ln q}\right)\]
  and
  \[ \sum_{i=1}^q \left[ 2 \mu_i^2 -
    \Lambda_{\frac{c}{q-1}}(\mu_i)\right]^+ \ge \sum_{i=1}^q (2\mu_i^2
  - F(\mu_i))\] which is convex because $\mu_i\le (1/q)^{1/10}$ and
  minimized at $\mu_i=1/q$. In this case, we have
  \[ \sum_{i=1}^q (2\mu_i^2 - F(\mu_i)) \ge q\left(q^{-2} - 2 c q^{-3}
    \ln q (1 + C \ln \ln q / \ln q\right)\] from which the above claim
  follows.
\end{proof}

\section{Analysis of PCP verifier for \maxcolor{k}}
\label{app:kcol-PCP}

\subsection{Proof of Lemma~\ref{lem:kcol-compl}}
\begin{proof}
  Let $\ell:V\to \intset{2 R}$ be a labeling for $\mathcal{L}$
  satisfying all constraints in $\Pi$. Pick $\chi_v$ as the Long Code
  encoding of $\ell(v)$. Given any pair of vertices $v,v' \in V$ which
  share a common neighbor $u \in U$, and $x,y\in [k]^{2 R}$ pairs such
  that 
\[ \prob{}{\overline{y} \sim T^{\otimes R}(\overline{x})} = \prod_i
  T(\xiffy{(x_{2i-1},x_{2i})}{(y_{2i-1},y_{2i})})>0 \ , \]
 let $\pi,\pi'$
  be the projection functions and $\sigma_v,\sigma_{v'}$ be the permutations
  as defined in the description of the verifier. We have $\chi_v(x \ \circ\ \sigma_v) =
  x_{\sigma(\ell(v))}$ and $\chi_{v'}(y \ \circ\ \sigma_{v'}) =
  y_{\sigma'(\ell(v'))}$. Since $\pi(\ell(v)) = \pi'(\ell(v'))$, this
  implies $\sigma_v(\ell(v)),\sigma_{v'}(\ell(v')) \in \{2i-1,2i\}$ for some
  $i\le R$. But 
\[ T(\xiffy{(x_{2i-1},x_{2i})}{(y_{2i-1},y_{2i})})>0
  \implies \{x_{2i-1},x_{2i}\}\cap \{y_{2i-1},y_{2i}\}=\emptyset \ , \]
  therefore $\chi_{v}\circ \sigma_v(x) = x_{\sigma_v(\ell(v))} \neq
  y_{\sigma_{v'}(\ell(v'))}=\chi_{v'}\circ \sigma_{v'}(y)$. So the verifier always accepts.
\end{proof}
\subsection{Proof of Lemma~\ref{lem:kcol-compl}}
\begin{proof}
  For each node $v \in V$, let $f^v:[k]^{2 R} \rightarrow \Delta_k$ be
  the function $f^v(x) = e_{\chi_v(x)}$ where $e_i$
is the indicator
  vector of the \texsup{$i$}{th} coordinate. Let $\Gamma(u)$ denote
  the set of vertices adjacent to $u$ in the Label Cover graph.

  After arithmetizing, we can write the verifier's acceptance probability as
%
  \[
    \begin{array}{rclr}
      \prob{}{\text{acc}} &=&
      \expct{u,v,v'} {1-
        \sum_j \langle \overline{f^v_j \circ\ \sigma_v},
        T^{\otimes R} \overline{(f^{v'}_j \circ\ \sigma_{v'})}\rangle} \\
      &=& 1 - \expct{u}{\sum_j \expct{v,v'} {
          \langle \overline{f^v_j\ \circ\ \sigma_v},
          T^{\otimes R} \overline{(f^{v'}_j\ \circ\ \sigma_{v'})}\rangle}} \\
      &=& 1 - \expct{u} {\sum_j
        \langle \expct{v}{\overline{f^v_j\ \circ\ \sigma_v}},
        T^{\otimes R} \expct{v'}{
          \overline{f^{v'}_j\ \circ\ \sigma_{v'}}}\rangle} \\
      &=& 1 - \expct{u} {
        \sum_j\langle g^u_j, T^{\otimes R} g^u_j\rangle} & 
      \left(g^u_j = \expct{v\sim \Gamma(u)}
                          {\overline{f^v_j\circ \sigma_v}}\right)\\
      &\ge& 1 - 1/k + C \ln k/k^2 
    \end{array}
  \] where $g^u:[k^2]^{R}\to \Delta_{k}$ and some constant $C$.
  By averaging, for at least a fraction
  $\delta=(\eps/2) \ln k/k^2$ of vertices in $U$, we have
   \[
  \sum_j\langle g^u_j, T^{\otimes R} g^u_j\rangle \le 1/k-C \ln
  k/k^2\]
  Let these be ``good'' vertices. For a good vertex, by Proposition
  \ref{prop:stab}, there exist constants $\delta=\delta(k)$, $t=t(k)$ and 
  $i$
  such that $\infl_i^{\le t}(g^u)\ge\delta$. Let $\sugg_u = \{i| i\in
  \intset{R}\wedge \infl_{i}^{\le t} (g^u)\ge \delta\}$, so $|\sugg_u|\ge
  1$. By Observation
  \ref{obs:infbndsimplex}, $|\sugg_u|\le t/\delta$. For a good vertex
  $u$, and $j\in\sugg_u$:
  \[ \delta \le \infl_j^{\le t} (g^u) = \expct{v \sim
    \Gamma(u)}{\infl_{j}^{\le t} \bigl( \overline{f^v \circ \sigma_v} \bigr)} \] 
  Therefore, for at least a fraction $\delta/2$ of neighbors $v$ of $u$, $\infl^{\le
    t}_{j}(\overline{f^v \circ \sigma_v}) \ge \delta/2$.  For such $v$
  and $j$, by Claim \ref{claim:infrel}, $\infl_{2j-1}^{\le 2t}
    (f^v\circ \sigma_v) + \infl_{2j}^{\le 2t} (f^v\circ \sigma_v) \ge
    \delta/2$. Therefore for some $j\in[2 R]$, $\infl_j^{\le 2t}
  (f^v) \ge \delta/4$. Let $\sugg_v = \{j| j\in \intset{2 R}\wedge
    \infl_{j}^{\le 2t} (f^v)\ge \delta/4\}$. Again, $\sugg_v$ is not
  empty and $|\sugg_v|\le 8 t/\delta$.

  Following the decoding procedure in \cite{KKMO}, we deduce that it
  is possible to satisfy a fraction $\gamma' = \gamma'(\delta,t) = \gamma'(k)$
  of the constraints.
\end{proof}

\section{Handling $k$ not divisible by $3$ in Theorem~\ref{thm:kcolor-NP-hard}}
\label{app:k-mod-3}

We now argue how to handle the case when $k \mod 3 \neq 0$ in the
statement of Theorem~\ref{thm:kcolor-NP-hard}. Assume $k$ is of the
form $K+L$, where $K\equiv 0\pmod 3$ and $L = k \mod 3 \in
\{1,2\}$. We will give a reduction from \maxcolor{K}, which we already
showed to be NP-hard to approximate within a factor
$1-\frac{1}{33K}+\eps$, to \maxcolor{k}.

Let $G_K$ be an (unweighted) instance of \maxcolor{K} with $M$
edges. Construct a graph $H$ by adding $L$ new vertices
$u_1,\dots,u_L$ to $G_K$. Each $u_i$ is connected by an edge of weight
$\frac{d_v}{K}$ to each vertex $v \in V(G_K)$, where $d_v$ is the
degree of $v$ in $G_K$. If $L > 1$, $(u_1,u_2)$ is an edge in $H$ with
weight $\frac{M}{33K}$. The total weight of edges in $H$ equals
\[ M' = M + \frac{2LM}{K} + \frac{M(L-1)}{33K} \ . \]

Clearly if $G_K$ is $K$-colorable, then $H$ is $k$-colorable. For the
soundness part, suppose every $K$-coloring of $G_K$ miscolors at least
$\Bigl(\frac{1}{33K}-\eps\Bigr) M$ edges. Let $\chi$ be an optimal
$k$-coloring of $H$. We will prove that $\chi$ miscolors edges with
total weight at least $M(\frac{1}{33K}-\eps)$. This will certainly be the case
if $L > 1$ and $\chi(u_1)=\chi(u_2)$. So we can assume $\chi$ uses $L$
colors for the newly added vertices $u_i$. If $\chi(v) = \chi(u_i)$ for some $v
\in V(G_K)$, we can change $\chi(v)$ to one of the $K$ colors not used
to color $\{u_1,\dots,u_L\}$ so that the weight of
miscolored edges does not increase. Therefore, we can assume that
$\chi$ uses only $K$ colors to color the $G_K$ portion of $H$. But
this implies at least $M(\frac{1}{33K}-\eps)$ edges are miscolored by
$\chi$, as desired.

Thus every $k$-coloring of $H$ miscolors at least a fraction
\[ \frac{M(1/(33K)-\eps)}{M'} = \frac{(1/(33K)-\eps)}{1 + 2L/K + (L-1)/(33K)} \ge \frac{1}{33(k+L)
  + (L-1)} -\eps \]
of the total weight of edges in $H$. Since $L = k \mod 3$, the bound
stated in Theorem~\ref{thm:kcolor-NP-hard} holds.


\end{document}